\def\arxiv{}
\documentclass[letterpaper]{article} 
\usepackage{aaai23}  
\usepackage{times}  
\usepackage{helvet}  
\usepackage{courier}  
\usepackage[hyphens]{url}  
\usepackage{graphicx} 
\ifdefined\arxiv
\usepackage[style=authoryear,maxbibnames=99]{biblatex}
\bibliography{results}
\else
\usepackage{natbib}  
\fi
\usepackage{caption} 
\usepackage
[
    ruled,
    vlined,
    linesnumbered,
]
{algorithm2e}

\ifdefined\arxiv
\newcommand{\citet}[1]{\textcite{#1}}
\renewcommand{\cite}[1]{\parencite{#1}}
\usepackage[hidelinks]{hyperref}
\else
\fi

\usepackage{etoolbox}
\makeatletter
\patchcmd{\@algocf@start}
  {-1.5em}
  {0pt}
  {}{}
\makeatother
\usepackage{amsmath}
\usepackage{amsthm}
\usepackage{amsfonts}
\usepackage{enumitem}
\usepackage{thm-restate}
\usepackage
[
    noabbrev, nameinlink
]
{cleveref}

\usepackage{comment}

\usepackage{mathrsfs}
\usepackage{tikz}
\usetikzlibrary{arrows.meta,calc}

\newtheorem{theorem}{Theorem}
\newtheorem{lemma}{Lemma}
\newtheorem{definition}{Definition}

\newtheorem{example}{Example}

\title{Strategic Facility Location with Clients That Minimize Total Waiting Time}
\author {
    Simon Krogmann\textsuperscript{\rm 1},
    Pascal Lenzner\textsuperscript{\rm 1},
    Alexander Skopalik\textsuperscript{\rm 2}
}
\affiliations {
    \textsuperscript{\rm 1} Hasso Plattner Institute, University of Potsdam\\
    \textsuperscript{\rm 2} Department of Applied Mathematics, University of Twente\\
    simon.krogmann@hpi.de, pascal.lenzner@hpi.de, a.skopalik@utwente.nl
}

\usepackage{todonotes}

\newcommand{\s}{\mathbf{s}}
\newcommand{\limtmpmodel}{$2$-FLG}
\newcommand{\brokenmodel}{Min-$2$-FLG}
\newcommand{\feldmanmodel}{Uniform-$2$-FLG}
\newcommand{\brokenshort}{\textup{Min}}
\newcommand{\feldmanshort}{\textup{Uniform}}

\ifdefined\arxiv
\nocopyright
\fi

\begin{document}

\maketitle

\begin{abstract}
We study a non-cooperative two-sided facility location game in which facilities and clients behave strategically.
This is in contrast to many other facility location games in which clients simply visit their closest facility.
Facility agents select a location on a graph to open a facility to attract as much purchasing power as possible, while client agents choose which facilities to patronize by strategically distributing their purchasing power in order to minimize their total waiting time. Here, the waiting time of a facility depends on its received total purchasing power.
We show that our client stage is an atomic splittable congestion game, which implies existence, uniqueness and efficient computation of a client equilibrium.
Therefore, facility agents can efficiently predict client behavior and make strategic decisions accordingly.
Despite that, we prove that subgame perfect equilibria do not exist in all instances of this game and that their existence is NP-hard to decide.
On the positive side, we provide a simple and efficient algorithm to compute 3-approximate subgame perfect equilibria.
\end{abstract}

\section{Introduction}
Facility location problems are a classical and popular object of study in AI, Operations Research, Economics, and Theoretical Computer Science. These problems encompass natural problems like placing facilities in a socially optimal way, e.g., the placement of hospitals, fire stations, or schools, or determining the locations of competing facilities, e.g., bars, shops, or gas stations. While the former type of problems are typically modeled as optimization problems and solved with a rich toolbox of combinatorial graph algorithms, this cannot be done for the latter problems, due to the strategic setting. Instead, models and methods from (Algorithmic) Game Theory are necessary to cope with facility location in competitive environments. The focus of this paper will be one such model that explicitly considers strategic behavior by the facilities as well as by the clients that want to patronize these facilities. Most importantly, we thereby consider a very natural type of client behavior that has, to the best of our knowledge, not yet been studied. 

The study of strategic facility location models dates back almost a century to the works of \citet{hotelling} who considered how to place two competing shops in a linear market, called the ``main street''.
Later, this model was refined by \citet{downs} and used for the placement of political candidates in a political left-right spectrum.
In the Hotelling-Downs model, competing but otherwise identical facilities strategically select a placement in some underlying space, e.g., on the line, to attract as many clients as possible.
The clients are modeled in a very simple way: clients are assumed to patronize their nearest opened facility.
However, while such simplistic clients are still considered in a wide range of models and recent works, such basic client agents do not seem to convincingly capture the behavior of real-world clients.
Abstracting away from more complex aspects like pricing, product quality, or social influences, realistic clients would not only consider the distance to a facility but also the incurred waiting time at the facility.
Interestingly, this waiting time depends on the number of other clients that also patronize the same facility.
Thus, in a more realistic setting we have that not only the facilities act strategically, but also the clients base their behavior on the behavior of the other clients instead of myopically minimizing distances.
Such more realistic clients have first been considered in the work of \citet{load-balancing}, who studied clients that minimize a convex combination of distance and waiting time.
However, this enhanced realism comes at a cost: while for the original Hotelling-Downs model equilibrium states always exist, except for exactly $3$ facilities, equilibria only exist for extreme cases of Kohlberg's model.

But, as demonstrated by~\citet{feldman-hotelling}, these negative results can be circumvented by considering a natural and interesting variant of the Hotelling-Downs model.
In this model, instead of minimizing the distance, every client has a threshold of how far she would travel to a facility and chooses to patronize a facility uniformly at random whose distance is below the threshold.
So essentially, the clients uniformly split their purchasing power among all facilities that are close enough.
Besides being more realistic than clients committing to a single facility, this model variant always has equilibrium facility placements.
However, it also has a downside: it considers simplistic clients that do not take waiting times into consideration.
Very recently, this downside was removed in a model proposed by~\citet{KLMS21}. In their two-stage facility location game, a given graph encodes which locations can reach other locations. Clients with a certain purchasing power occupy the nodes of this graph and facility agents strategically select a node for opening their facilities. Given a facility placement, the clients then split their purchasing power among the facilities they can reach to minimize their maximum waiting time. Analogously to classical make-span scheduling, these clients effectively try to distribute their purchasing power to load-balance the reachable facilities. 

In this paper, we set out to explore a similar setting as \citet{KLMS21} but with a drastically different client behavior. To stay in the analogy with scheduling, we consider sum-of-completion-time scheduling instead of make-span scheduling. I.e., we study facility location problems with clients that distribute their purchasing power to minimize their total waiting time.

\paragraph*{Related Work.} There is an abundance of models for strategic facility location and we refer to the surveys by \citet{ELT93} and \citet{RE05} for an overview of the classical models. To the best of our knowledge, our model has not yet been studied. 

The model by \citet{feldman-hotelling} with distance-minimizing clients on a line was generalized by~\citet{hotelling-limited-attraction} to cope with different continuous client distributions and by \citet{hotelling-line-bubble} to random distance thresholds. Moreover,
 \citet{fournier2020spatial} study clients that choose the nearest facility if they have multiple options. Also related are Voronoi games~\cite{voronoi1d,voronoigames}, which can be understood as generalized Hotelling-Downs models with distance-minimizing clients. For the version on networks by \citet{voronoigames}, the authors show that equilibria may not exist and that existence is NP-hard to decide. Also a variant on a cycle~\cite{voronoi-cycle} and in $k$-dimensional space~\cite{voronoi-kd-voters,voronoi1d,voronoi-choice} was studied.
 
For Kohlberg's model, \citet{Peters2018} prove the existence of equilibria for certain trade-offs of distance and waiting time for small even numbers of facilities and they conjecture that equilibria exist for all cases with an even number of facilities for client utility functions that are heavily tilted towards minimizing waiting times. \citet{hotelling-load-balancing} showed that computing equilibria for Kohlberg's model can be done by solving a complex system of equations and they investigated the existence of approximate equilibria. They find that under a technical assumption, $1.08$-approximate equilibria always exist.

A concept related to our model are utility systems, as introduced by \citet{vetta-utility-system}. There, agents gain utility by selecting a~set of acts, which they choose from a collection of subsets of a groundset. Utility is assigned by a function that takes the selected acts of all agents as input. 
Also covering games~\cite{gairing-covering-games} are related since they correspond to a one-sided version of our model, where clients distribute their purchasing power uniformly among all facilities in reach. In both settings, utility systems and covering games, pure NE exist and the Price of Anarchy is upper bounded by~$2$. More general versions are investigated by \citet{3g-market-sharing} and \citet{alex-market-sharing} in the form of market sharing games. In these models, $k$ agents choose to serve a subset of $n$ markets. Each market then equally distributes its utility among all agents serving it. \citet{alex-network-investment} introduced a model which considers an inherent load balancing problem, however, each facility agent can create and choose multiple facilities and each client agent chooses multiple facilities.
An empirical investigation of a two-sided facility location problem was conducted by \citet{empirical-2-stage}, in which a single facility agent opens facilities for strategic clients.
The facility agent sets service levels while the client agents determine their strategies based on these levels, but also travel distance and congestion.

Closest to our work is the above-mentioned work by~\citet{KLMS21}, where clients that minimize their maximum waiting time are considered. For this version equilibria always exist due to a potential function argument. Moreover, the authors present a polynomial time algorithm for computing the unique client equilibrium, given a facility placement. In terms of quality of the equilibria, the authors prove an essentially tight bound of 2 on the Price of Anarchy by establishing that this model is a valid utility system.
These results on quality apply to a general class of games which also includes the game we study in this paper.

Facility location problems were also studied recently with Mechanism Design, e.g.,~\citet{PT13,FFG16,ACL0W20,CFLLW21,HLST21,FP21,DEG22}.

\paragraph*{Our Contribution.}
In this paper, we introduce a two-sided facility location game with clients that minimize their total incurred waiting time.
In the first stage, each facility agent individually selects a location on a graph to open a facility, while in the second stage client agents choose which facilities to patronize.
Thereby, clients may freely distribute their purchasing power among all opened facilities within their shopping range.
In contrast to a similar game by \citet{KLMS21}, our clients minimize their total expected waiting time instead of their maximum expected waiting time.
To the best of our knowledge, we are the first to study this natural behavior in a facility location game.

The client behavior on its own is well studied in the form of \emph{atomic splittable congestion games} and we reduce our client stage to variants of these games to show existence, uniqueness and polynomial time computation of client equilibria.
This means that facilities can efficiently predict client behavior to inform their strategic location decisions.

However, on the negative side, we show that for our complete game consisting of both the facility stage and the client stage, subgame perfect equilibria (SPE) are not guaranteed to exist for all instances. This is in surprising contrast to the similar game studied by \citet{KLMS21}, which is actually a potential game with guaranteed equilibrium existence.
In fact, we prove that it is even NP-hard to decide the existence of an SPE for a given instance. This shows that the specific client behavior has a severe impact on the obtained game-theoretic properties. 

But, on the positive side, we prove the existence of 3-approximate SPEs and show that they can be computed efficiently by using another facility location game as a proxy.

\section{Model and Preliminaries}

We consider a game-theoretic model for non-cooperative facility location, called the \emph{Two-Sided Facility Location Game (\limtmpmodel{})}\footnote{We reuse notation by \citet{KLMS21}.}, where two types of agents, $k$ \emph{facilities} and $n$ \emph{clients}, strategically interact on a given vertex-weighted directed\footnote{Our results also hold for undirected graphs. Also, our model can be equivalently defined on an undirected bipartite graph.} host graph $H = (V,E,w)$, with $V = \{v_1,\dots,v_n\}$, where $w:V \to \mathbb{Q}_+$ denotes the vertex weight.
Every vertex $v_i \in V$ corresponds to a client with weight $w(v_i)$, which can be understood as her purchasing power, and at the same time, each vertex is a possible location for setting up a facility for any of the $k$ facility agents $\mathcal{F} = \{f_1,\dots,f_k\}$. Any client $v_i \in V$ considers patronizing a facility in her \emph{shopping range} $N(v_i)$, i.e., her direct closed neighborhood $N(v_i) = \{v_i\} \cup \{z \mid (v_i,z) \in E\}$. Moreover, let $w(X) = \sum_{v_i \in X}w(v_i)$, for any $X \subseteq V$, denote the total purchasing power of the client subset $X$.

In our setting, the strategic behaviors of the facility and the client agents influence each other. Facility agents select a location to attract as much client weight, i.e., purchasing power, as possible, whereas clients strategically decide how to distribute their purchasing power among the facilities in their respective shopping ranges.
More precisely, each facility agent $f_j \in \mathcal{F}$ selects a single location vertex $s_j \in V$ for setting up her facility, i.e., the strategy space of any facility agent $f_j\in \mathcal{F}$ is $V$. Let $\s = (s_1,\dots,s_k)$ denote the \emph{facility placement profile}. And let $\mathcal{S} = V^k$ denote the set of all possible facility placement profiles.
We will sometimes use the notation $\s = (s_j,s_{-j})$, where $s_{-j}$ is the vector of strategies of all facilities agents except $f_j$.
Given profile $\s$, we define the \emph{attraction range} for a facility $f_j$ on location $s_j \in V$ as $A_\s(f_j) = \{s_j\} \cup \{v_i \mid (v_i,s_j) \in E\}$. We extend this to sets of facilities $F \subseteq \mathcal{F}$ in the natural way, i.e., $A_\s(F) = \{s_j \mid f_j \in F\} \cup \{v_i \mid (v_i,s_j) \in E, f_j\in F\}$. Moreover, let $w_\s(\mathcal{F}) = \sum_{v_i \in A_\s(\mathcal{F})} w(v_i)$.

We assume that all facilities provide the same service for the same price and arbitrarily many facilities may be co-located at the same location.
Each client $v_i\in V$ strategically decides how to distribute her purchasing power $w(v_i)$ among the opened facilities in her shopping range~$N(v_i)$. For this, let $N_\s(v_i) = \{f_j \mid s_j \in N(v_i)\}$ denote the set of facilities in the shopping range of client $v_i$ under $\s$.

Let $\sigma: \mathcal{S} \times V \to \mathbb{R}_+^k$ denote the \emph{client weight distribution function}, where $\sigma(\s,v_i)$ is the weight distribution of client $v_i$ and $\sigma(\s,v_i)_{j}$ is the weight distributed by $v_i$ to facility~$f_j$. We say that function $\sigma$ is \emph{feasible} for profile $\s$, if all clients having at least one facility within their shopping range distribute all their weight to the respective facilities and all other clients distribute nothing.
Formally, function $\sigma$ is feasible for profile $\s$, if for all $v_i\in V$ we have $\sum_{f_j \in N_\s(v_i)} \sigma(\s,v_i)_j = w(v_i)$, if $N_\s(v_i) \neq \emptyset$, and $\sigma(\s,v_i)_j = 0$, for all $1\leq j \leq k$, if $N_\s(v_i) = \emptyset$.
We use the notation $\sigma = (\sigma_i,\sigma_{-i})$ where $(\sigma_i',\sigma_{-i})$ denotes the changed client weight distribution function that is identical to $\sigma$ except for client $v_i$, which plays $\sigma'(\s,v_i)$ instead of $\sigma(\s,v_i)$.

Any state $(\s,\sigma)$ of the \limtmpmodel{} is determined by a facility placement profile~$\s$ and a feasible client weight distribution function $\sigma$.
A state $(\s,\sigma)$ then yields a \emph{facility load} $\ell_j(\s,\sigma)$, with $\ell_j(\s,\sigma) = \sum_{i=1}^n \sigma(\s,v_i)_j$ for facility agent $f_j$. Hence, $\ell_j(\s,\sigma)$ naturally models the total congestion for the service offered by the facility of agent $f_j$ induced by~$\sigma$. A facility agent $f_j$ strategically selects a location $s_j$ to maximize her induced facility load $\ell_j(\s,\sigma)$. We assume that the service quality of facilities, e.g., the waiting time, deteriorates with increasing congestion. Hence for a client, the facility load corresponds to the waiting time at the respective facility.

There are many ways how clients could distribute their purchasing power.
We investigate the \emph{\brokenmodel{}} with
\emph{waiting time minimizing clients}, i.e., a natural strategic behavior where client $v_i$ strategically selects $\sigma(\s,v_i)$ to minimize her total waiting time.
More precisely, the cost of client $i$ is 
$$L_i(\s,\sigma) = \sum_{j=1}^k \sigma(\s,v_i)_j \ell_j(\s,\sigma)\text.$$
Another version of the \limtmpmodel{} we use in this paper is the \emph{\feldmanmodel{}}, in which clients distribute their weight equally among all facilities in their range.
Formally, for each pair of client $v_i$ and facility $f_j$, with $f_j \in N_\s(v_i)$, the client's weight is $\sigma(\s, v_i)_j = \frac{w(v_i)}{N_\s(v_i)}$.
Another such model is the load balancing \limtmpmodel{} introduced by \citet{KLMS21}, which we do not consider in this paper.

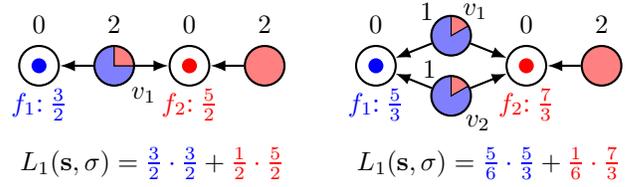
\begin{figure}
    \centering
    \begin{tikzpicture}
        \clip (-.4, -1.56) rectangle (3.4, 0.86);
        \coordinate(splitnode) at (1, 0);
        \coordinate(wholenode) at (3, 0);
        \filldraw[fill=blue!50] ($(splitnode)+(0,7.5pt)$) arc (90:360:7.5pt) -- (splitnode) -- cycle;
        \filldraw[fill=red!50] ($(splitnode)+(0,7.5pt)$) arc (90:0:7.5pt) -- (splitnode) -- cycle;
        \filldraw[fill=red!50, draw=none] ($(wholenode)+(0,7.5pt)$) arc (90:450:7.5pt);
        \begin{scope}[every node/.style = {circle, thick, draw, inner sep = 0pt, minimum size = 15pt},every path/.style = {thick, {-Latex[length=2mm]}}]
            \node (0) [label=above:$0$, label={[blue, label distance=-0.2cm]below:$f_1$: $\frac32$}] at (0, 0) {};
            \node (1) [label=above:$2$] [label=below right:$v_1$] at (splitnode) {};
            \node (2) [label=above:$0$, label={[red, label distance=-0.2cm]below:$f_2$: $\frac52$}] at (2, 0) {};
            \node (3) [label=above:$2$] at (wholenode) {};
            \draw (1) edge (0) edge (2) (3) edge (2);
        \end{scope}
        \tikzstyle{fac}=[circle, inner sep = 2pt, fill]
        \node (f1)[fac, blue] at (0) {};
        \node (f2)[fac, red] at (2) {};
        \node at (1.5, -1.3) {$L_1(\s,\sigma)= {\color{blue}\frac32\cdot\frac32}+{\color{red}\frac12\cdot\frac52}$};
    \end{tikzpicture}
    \hspace{5mm}
    \begin{tikzpicture}
        \clip (-.4, -1.56) rectangle (3.4, 0.86);
        \coordinate(splitnodea) at (1, 0.4);
        \coordinate(splitnodeb) at (1, -0.4);
        \coordinate(wholenode) at (3, 0);
        \filldraw[fill=blue!50] ($(splitnodea)+(0,7.5pt)$) arc (90:390:7.5pt) -- (splitnodea) -- cycle;
        \filldraw[fill=red!50] ($(splitnodea)+(0,7.5pt)$) arc (90:30:7.5pt) -- (splitnodea) -- cycle;
        \filldraw[fill=blue!50] ($(splitnodeb)+(0,7.5pt)$) arc (90:390:7.5pt) -- (splitnodeb) -- cycle;
        \filldraw[fill=red!50] ($(splitnodeb)+(0,7.5pt)$) arc (90:30:7.5pt) -- (splitnodeb) -- cycle;
        \filldraw[fill=red!50, draw=none] ($(wholenode)+(0,7.5pt)$) arc (90:450:7.5pt);
        \tikzset{every path/.style = {thick, {-Latex[length=2mm]}}}
        \begin{scope}[every node/.style = {circle, thick, draw, inner sep = 0 pt, minimum size = 15 pt}]
            \node (0) [label=above:$0$, label={[blue, label distance=-0.2cm]below:$f_1$: $\frac53$}] at (0, 0) {};
            \node (1a) [label={[label distance=-0.1cm]above left:$1$}] [label={[label distance=-0.1cm]above right:$v_1$}] at (splitnodea) {};
            \node (1b) [label={[label distance=-0.1cm]120:$1$}] [label={[label distance=-0.05cm]below right:$v_2$}] at (splitnodeb) {};
            \node (2) [label=above:$0$, label={[red, label distance=-0.2cm]below:$f_2$: $\frac73$}] at (2, 0) {};
            \node (3) [label=above:$2$] at (wholenode) {};
            \draw (1a) edge (0) edge (2) (1b) edge (0) edge (2) (3) edge (2);
        \end{scope}
        \tikzstyle{fac}=[circle, inner sep = 2pt, fill]
        \node (f1)[fac, blue] at (0) {};
        \node (f2)[fac, red] at (2) {};
        \node at (1.5, -1.3) {$L_1(\s,\sigma)= {\color{blue}\frac56\cdot\frac53}+{\color{red}\frac16\cdot\frac73}$};
    \end{tikzpicture}
    \caption{
        Two instances of the \brokenmodel{} with their client equilibria visualized for a given facility placement profile.
        The clients on each node split their weight (above the nodes) among the facilities in their shopping ranges to minimize their cost.
        The facilities are marked by dots inside the nodes and receive the loads below in the respective unique client equilibria.
        The respective cost $L_1(\s,\sigma)= \sigma(\s,v_1)_1 \ell_1(\s,\sigma) + \sigma(\s,v_1)_2 \ell_2(\s,\sigma)$ of client $v_1$ (and also $v_2$ on the right) is given below each instance.
    }
    \label{fig:example}
\end{figure}

We say that $\sigma^*$ is a \emph{client equilibrium weight distribution}, or simply a \emph{client equilibrium}, if for all $v_i \in V$ and for all placement profiles $\s$ we have that $L_i(\s,(\sigma_{i}^*,\sigma_{-i})) \leq L_i(\s,(\sigma_{i}',\sigma_{-i}))$ for all feasible client weight distribution functions $\sigma'$. See \Cref{fig:example} for an illustration of the client behavior in the \brokenmodel{}.

We define the \emph{stable states} of the \limtmpmodel{} as \emph{subgame perfect equilibria (SPE)}, since we inherently have a two-stage game. First, the facility agents select locations for their facilities and then, given this facility placement, the clients strategically distribute their purchasing power among the facilities in their shopping range.
A state $(\s,\sigma)$ is in SPE, or \emph{stable}, if
\begin{enumerate}[label=(\arabic*), align=left]
    \item $\forall f_j \in \mathcal{F},  \forall s_j' \in V$: $\ell_j(\s,\sigma) \geq \ell_j((s_j',s_{-j}),\sigma)$ and
    \item $\forall \s' \in \mathcal{S}, \forall v_i \in V$: $L_i(\s',\sigma) \leq L_i(\s',(\sigma_{i}',\sigma_{-i}))$ for all feasible client weight distribution functions $\sigma'$.
\end{enumerate}
As we will show, SPE do not always exist. Hence,  we relax the first condition as follows  and obtain the
notion of $\alpha$-\emph{approximate subgame perfect equilibria ($\alpha$-SPE)}:
\begin{enumerate}[label=(\arabic*'), align=left]
    \item $\forall f_j \in \mathcal{F},  \forall s_j' \in V$: $\ell_j(\s,\sigma) \geq \alpha \ell_j((s_j',s_{-j}),\sigma)$ 
   
\end{enumerate}



We study dynamic properties of the \limtmpmodel{}.
Let an \emph{improving move} by some (facility or client) agent be a strategy change that improves the agent's utility.
A game has the \emph{finite improvement property (FIP)} if all sequences of improving moves are finite.
The FIP is equivalent to the existence of an \emph{ordinal potential function}~\cite{MS96}, which implies equilibrium existence.

\section{Client Equilibria}

The existence of client equilibria is implied by Kakutani's fixed-point theorem~\cite{kakutani}.\ifdefined\arxiv\footnote{See \Cref{sec:existence} for the proof.}\fi\ 
A reduction of our game to \emph{atomic splittable routing games} by \citet{bhaskar2015uniqueness} obtains uniqueness of client equilibria.\ifdefined\arxiv\footnote{See \Cref{sec:uniqueness} for the reduction to \emph{atomic splittable routing games}.}\fi\ Note that for the reduction to their model to work, we require the extension defined by the authors where each player may have her own individual delay function for each edge.
\begin{restatable}{theorem}{uniqueness}
\label{theo:uniqueness}
For a given facility placement profile $\s$ the client equilibrium in the \brokenmodel{} is unique.
\end{restatable}
Client equilibria can be computed in polynomial time with an algorithm given by \citet{harks_equilibrium_2021}.\ifdefined\arxiv\footnote{This requires a reduction to \emph{atomic splittable singleton congestion games} given in \Cref{sec:computation}.}\fi
\begin{restatable}{theorem}{computation}
\label{theo:polytime-client}
For a given facility placement profile $\s$ the client equilibrium in the \brokenmodel{} can be computed in polynomial time.
\end{restatable}
Notably, the most involved step to compute a client equilibrium is to determine the pairs of clients and facilities with non-zero weight.
This information then yields a set of linear equations solvable by Gaussian elimination.

\ifdefined\arxiv\else
We refer to \citet{arxiv} for details of the application of Kakutani's fixed-point theorem and the reductions for \Cref{theo:uniqueness,theo:polytime-client}.
\fi

We now show how the loads of two facilities with a shared client relate to each other in a client equilibrium.

\begin{lemma}
\label{lemma:shared-client}
Let $f_p$ and $f_q$ be two facilities that share a client $v_i$ in a given facility placement profile $\s$.
In a client equilibrium $\sigma$, if $\sigma(\s, v_i)_p > 0$, then
\[
    \ell_p(\s, \sigma) + \sigma(\s, v_i)_p \leq \ell_q(\s, \sigma) + \sigma(\s, v_i)_q\text.
\]
\end{lemma}
\begin{proof}
If $v_i$ transfers weight from $f_p$ to $f_q$ it only affects the terms $\sigma(\s, v_i)_p \ell_p(\s, \sigma)$ and $\sigma(\s, v_i)_q \ell_q(\s, \sigma)$ in her cost and not the other terms of the sum.
Let $L_{i,j} = \sigma(\s, v_i)_j \ell_j(\s, \sigma)$, for all $v_i \in V$ and $f_j \in \mathcal{F}$.
Since $\sigma$ is the unique client equilibrium, we know that for a transfer of weight of $\epsilon$ with $0 < \epsilon \leq \sigma(\s, v_i)_p$ we have
$L_{i,p} + L_{i,q}  < (\sigma(\s, v_i)_p - \epsilon) (\ell_p(\s, \sigma)-\epsilon) + (\sigma(\s, v_i)_q + \epsilon)(\ell_q(\s, \sigma) + \epsilon)$.
This yields $\ell_p(\s, \sigma) + \sigma(\s, v_i)_p < 2\epsilon + \ell_q(\s, \sigma) + \sigma(\s, v_i)_q$.
Since $\epsilon$ may be arbitrarily small, but not zero, this finishes the proof.
\end{proof}

\Cref{lemma:shared-client} also implies equality of the two terms, i.e., $\ell_p(\s, \sigma) + \sigma(\s, v_i)_p = \ell_q(\s, \sigma) + \sigma(\s, v_i)_q$, if a client $v_i$ has non-zero weight on both $f_p$ and $f_q$.

\section{Subgame Perfect Equilibria}

Unlike other versions of the \limtmpmodel{}, the \brokenmodel{} does not admit subgame perfect equilibria in all instances.
In fact, the counterexample only needs two facility agents with the host graph $H$ being a path of size 4.
Therefore, it may be complicated to find non-trivial subclasses of the \brokenmodel{} which always admit SPE.

\begin{theorem}
\label{theo:no-spe}
There are instances of the \brokenmodel{} for which a subgame perfect equilibrium does not exist.
\end{theorem}
\begin{proof}
Let $G^*=(V^*,E^*)$ with $V^*=\{v_1, v_2, v_3, v_4\}$ and $E^*=\{(v_1, v_2), (v_2, v_3), (v_3, v_4)\}$.
Let the client weights be $w(v_1)=3$, $w(v_2)=2$, $w(v_3)=7$, and $w(v_4)=1$.
For two facility agents, this instance does not admit an SPE.
\begin{figure}
    \centering
    \newcommand{\graphsetup}{
        \clip (-.55, -.55) rectangle (7.5, .55);
        \tikzset{every path/.style = {thick, {-Latex[length=2mm]}}}
        \begin{scope}[every node/.style = {circle, thick, draw, inner sep = 0pt, minimum size = 15pt}]
            \node (0) [label=above left:$3$] at (0, 0) {};
            \node (1) [label=above left:$2$] at (1, 0) {};
            \node (2) [label=above left:$7$] at (2, 0) {};
            \node (3) [label=above left:$1$] at (3, 0) {};
            \draw (0) edge (1) (1) edge (2) (2) edge (3);
        \end{scope}
        \tikzstyle{fac}=[circle, inner sep = 2pt, fill]
        \tikzstyle{improve}=[bend right, dashed]
        \tikzstyle{explain}=[anchor=west, align=left]
    }
    \newcommand{\explain}[4]{
        \node [anchor=east] at (5.3, 0) {${\color{blue}#1}, {\color{red}#2}$};
        \node at (5.5, 0) {$\to$};
        \node [anchor=west] at (5.7, 0) {${\color{blue}#3}, {\color{red}#4}$};
    }
    \begin{tikzpicture}
        \graphsetup{}
        \node (f1)[fac, blue] at (1) {};
        \node (f2)[fac, red] at (2) {};
        \draw (f2) edge[improve, red] (3);
        \explain{5}{7}{5}{8}
    \end{tikzpicture}
    \begin{tikzpicture}
        \graphsetup{}
        \node (f1)[fac, blue] at (1) {};
        \node (f2)[fac, red] at (3) {};
        \draw (f1) edge[improve, blue] (2);
        \explain{5}{8}{5.25}{4.75}
    \end{tikzpicture}
    \begin{tikzpicture}
        \graphsetup{}
        \node (f1)[fac, blue] at (2) {};
        \node (f2)[fac, red] at (3) {};
        \draw (f2) edge[improve, bend left, red] (1);
        \explain{5.25}{4.75}{7}{5}
    \end{tikzpicture}
    \begin{tikzpicture}
        \graphsetup{}
        \node (f1)[fac, blue, yshift=1.2mm] at (2) {};
        \node (f2)[fac, red, yshift=-1.2mm] at (2) {};
        \draw (f2) edge[improve, red] (3);
        \explain{4.5}{4.5}{5.25}{4.75}
    \end{tikzpicture}
    \caption{
        Best responses (dashed) to game states for the instance $G^*$ of the \brokenmodel{} without SPE.
        The utilities of the facilities before and after the move are given on the right.}
    \label{fig:no-spe}
\end{figure}
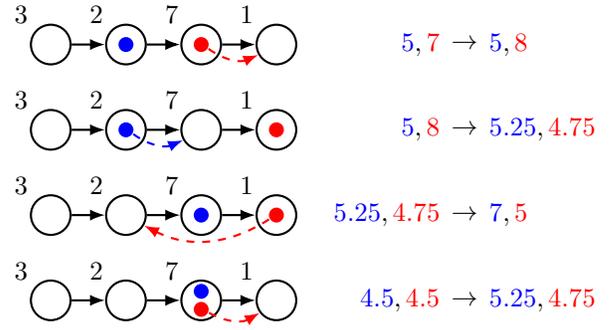
In \Cref{fig:no-spe} we show the best responses for states that cannot trivially be excluded as SPE.
\end{proof}
 Additionally,  determining whether an instance admits an SPE is computationally  intractable.

\begin{theorem}
\label{theo:NP}
Deciding if an instance of the \brokenmodel{} admits an SPE is NP-hard.
\end{theorem}
\begin{proof}
We reduce from {\sc IndependentSet (IS)}: Given a graph $G=(V,E)$ and an integer $k\le |V|$, decide whether there exists a subset $I\subseteq V$ with $|I| = k$ such that no two vertices in $I$ share an edge.  This problem is NP-hard, even for graphs with maximum degree $3$~\cite{GareyJohnson}.  In  the following,  we assume that $G$ has a maximum degree of at most $3$.

To prove the theorem, we construct an instance of the \brokenmodel{} on a host graph $H=(V',E',w)$ with $k'=2k$ facilities such that there is an SPE in $H$ if and only if $G$ contains an independent set of size $k$.
We obtain $H$ from $G$ by replacing every edge $e=\{u,v\}$ by a new vertex $x_e$ and two new edges $(x_e,u)$ and $(x_e,v)$.
For every vertex $v$ of degree $1$ (or $2$)  add two (or one) vertices $y_v$ (and $z_v$) and the edge $(y_v,v)$ (and $(z_v,v)$).
The newly added $x$-, $y$- and $z$-vertices have weight $1.75$,  vertices that originally belonged to $V$ have weight $0$.
Hence,  we now have in $H$ exactly $|V|$ many vertices that all have weight $0$ with exactly $3$ neighbors that have weight $1.75$ each.
All other vertices have weight $1.75$ and one or two neighbors with weight $0$.
To complete the construction of $H$,  we add $k$ copies of the graph $G^*$ that we used in \Cref{theo:no-spe} (see \Cref{fig:no-spe}).

Now, if $G$ contains an independent set $I$ of size $k$,  then there is an equilibrium in which $k$ facilities are placed on the vertices of $I$ and one facility each is placed on the second vertex from the right in each of the $k$ copies of $G^*$.
The first $k$ facilities each have a payoff of $5.25$ and,  hence, play their best response.  In particular, it is not an improvement to choose any vertex in any of the copies of $G^*$ as this yields at most $4.75$.
Each of the $k$ players in the copies of $G^*$ has a payoff of $9$ and, hence,  is clearly playing the best response.

If there is no independent set of size $k$ in $G$, then there cannot be more than $k-1$ players with a payoff of more than $4.375$ on vertices outside of the copies of $G^*$.
Note that a higher payoff is only possible on nodes of degree $3$ with no other facility within distance $2$.
However,  in an equilibrium no facility would choose a location with a payoff of at most $4.375$ as there is at least one of the copies of $G^*$ with at most one other facility on it.
Switching to the best vertex in that copy of $G^*$ guarantees a payoff of at least $4.5$.
Finally,  we observe that there is no equilibrium with two players on host graph $G^*$ in \Cref{fig:no-spe},  hence the is no equilibrium in $H$.
\end{proof}

\section{Approximation of SPE}

In this section, we show that an SPE in the \feldmanmodel{} is a 3-approximate SPE in the \brokenmodel{}.
We first give an example instance, where in an SPE in the \feldmanmodel{} a facility can improve by a factor of 2 when treating it as a state of the \brokenmodel{}.
This instance serves as a lower bound of the approximation quality using the \feldmanmodel{} and is also shown in \Cref{fig:bad-approx}.

\begin{example}
\label{ex:bad-approx}
Let $t > 0$ be a natural number.
Let $G=(V,E)$, with $V=\{v_b, v_a\} \cup \bigcup_{i=1}^t{V_i}$, with $V_i = \{x_i\} \cup \{y_{i,1}, \dots, y_{i,t}\}$ and $E = \bigcup_{i=1}^t{E_i}$, with $E_i = \{(v_a, x_i)\} \cup \{(x_i, y_{i,1}), \dots, (x_i, y_{i, t})\}$.
Let the client weights be $w(v_b)= 1$, $w(v_a)= 0$, $w(x_i) = 1$, for all $i$ and $w(y_{i,j}) = \frac{2t-2}{t}$, for all $i, j$.
Let the number of facilities be $k=t^2+1$.
\end{example}
For this example in a client equilibrium of the \feldmanmodel{}, facilities are located at $v_b$ and all of the nodes $y_{i,j}$, with $1 \leq i, j \leq t$, while for the \brokenmodel{} the facility agent on $v_b$ can improve by a factor of $2$ by moving to location $v_a$.

\begin{theorem}
Let $\s$ be an SPE for an instance of the \feldmanmodel{}. When transferred to the \brokenmodel{}, it is possible for profile $\s_\feldmanshort{}$ to be a $2$-approximate SPE.
\end{theorem}
\begin{proof}
For the \feldmanmodel{}, the facility placement profile $\s = (v_b, y_{1,1}, \dots y_{1,t}, \dots, y_{t,1}, \dots y_{t,t})$ is an SPE for \Cref{ex:bad-approx}, since the facility on $v_b$ receives $\frac{2t-2}{2t}+\frac{1}{t+1} < 1$ by switching to any $y_{i,j}$, receives $\frac{1}{t+1} < 1$ by switching to any $x_i$, and receives $\frac{1}{t+1}$ by switching to $v_a$.
Clearly, the other facilities lose more utility by switching strategies.

We transfer $\s$ to the \brokenmodel{}, after which the facility $f_j$ on $v_b$ still receives a utility of $1$.
We consider a switch by $f_j$ to $v_a$ resulting in the facility placement profile $\s'$.
Since there is only one client equilibrium, by symmetry all clients $x_1, \dots, x_t$ have the same cost and all facilities except $f_j$ have the same utility, so we fix some arbitrary client $v_i \in {x_1, \dots, x_t}$ and some facility $f_p \neq f_j$ with $f_p \in N_\s(x_i)$.

The cost of $v_i$ is $L_i(\s', \sigma)
= \sigma(\s', v_i)_j\ell_j(\s', \sigma) + \sigma(\s', v_i)_p\ell_p(\s', \sigma)t
= \sigma(\s', v_i)_j(\sigma(\s', v_i)_j +z)+\frac{1 - \sigma(\s', v_i)_j}{t}\left(\frac{1 - \sigma(\s', v_i)_j}{t}+\frac{2t-2}{t}\right)t$, where $z$ is the sum of the weight that $f_j$ receives from all clients except $v_i$.
Because the cost is minimal, we know that for the derivative $\frac{d}{d\sigma(\s', v_i)_j}L_i(\s', \sigma) = 0$ if and only if the minimum of $L_i(\s', \sigma)$ has $\sigma(\s', v_i)_j \in [0, w(v_i)]$.
Thus, we get $2\sigma(\s', v_i)_j + z-2+2\frac{\sigma(\s', v_i)_j}{t} = 0$.
We substitute $z = (t-1) \sigma(\s', v_i)_j$ because of symmetry\footnote{Note that we cannot do this substitution earlier, because doing it before applying the derivative would minimize the sum of utilities of all clients in $\{v_1,\dots, v_t\}$, instead of just $v_i$.} and then get $\sigma(\s', v_i)_j =\frac{2t}{t^2+t+2}$, which is within $[0, w(v_i)]$.
Thus, $L_j(\s', \sigma) = \frac{2t^2}{t^2+t+2}$, with $\lim_{t\to \infty} L_j(\s', \sigma) = 2$.
\end{proof}

\begin{figure}
\centering
\begin{tikzpicture}
    \clip (-4, -2.7) rectangle (4, 0.1);
    \node at (0, -1.6) {\dots};
    \node at (-2, -2) {\dots};
    \node at (2, -2) {\dots};
    \tikzset{every path/.style = {thick, {Latex[length=2mm]}-}}
    \tikzset{every node/.style = {circle, thick, draw, inner sep = 0pt, minimum size = 15pt}, text width=15pt, align=center}
    
    \node (before) [label={[label distance=-0.1cm]above left:$1$}] at (-3.5, -0.4) {$v_b$};
    \node (c) [label={[label distance=-0.1cm]above left:$0$}] at (0, -0.4) {$v_a$};
    \node (1a0) [label=above left:$1$] at (-2, -1.2) {$x_1$};
    \node (1b1) [label=below right:$\frac{2t-2}{t}$] at (-1, -2) {$y_{1,t}$};
    \node (1c1) [label=below right:$\frac{2t-2}{t}$] at (-3, -2) {$y_{1,1}$};
    \node (2a0) [label=above right:$1$] at (2, -1.2) {$x_t$};
    \node (2b1) [label=below right:$\frac{2t-2}{t}$] at (1, -2) {$y_{t,1}$};
    \node (2c1) [label=below right:$\frac{2t-2}{t}$] at (3, -2) {$y_{t,t}$};
    \draw (c) edge (1a0) edge (2a0)
        (1b1) edge (1a0) (1c1) edge (1a0)
        (2b1) edge (2a0) (2c1) edge (2a0);
\end{tikzpicture}
\caption{An instance of the \limtmpmodel{} for which an SPE in the \feldmanmodel{} is a $2$-approximate SPE in the \brokenmodel{}.
The improving facility improves by moving from $v_b$ to $v_a$.
}
\label{fig:bad-approx}
\end{figure}
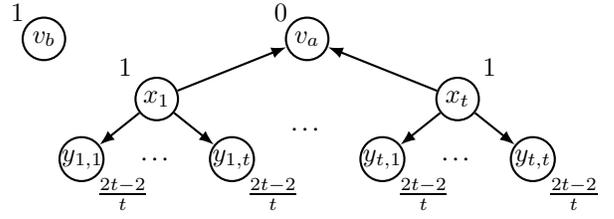

Next, we prove an upper bound on the approximation quality:
First, we show that of all facilities in range of a client $v_i$, in a client equilibrium for the \brokenmodel{}, the one with the lowest total load receives at least as much weight from $v_i$ as any other facility.
\begin{lemma}
\label{lemma:min-receives-more}
In the \brokenmodel{}, given a facility placement profile $\s$ and a client equilibrium $\sigma$, let $f_j$ be the facility in the attraction range $N_\s(v_i)$ of client $v_i$ with the lowest facility load. Then $\sigma(\s, v_i)_j \geq \sigma(\s, v_i)_x$ for any other facility $f_x \in N_\s(v_i)$.

Similarly, for $f_p \in N_\s(v_i)$ with the highest facility load: $\sigma(\s, v_i)_p \leq \sigma(\s, v_i)_x$ for any other facility $f_x$.
\end{lemma}
\begin{proof}
Let $f_x \in N_\s(v_i)$ be an arbitrary facility with non-zero weight $\sigma(\s, v_i)_x > 0$ on $v_i$.
Using \Cref{lemma:shared-client}, we get
\begin{align*}
\ell_j(\s, \sigma) + \sigma(\s, v_i)_j &\geq \ell_x(\s, \sigma) + \sigma(\s, v_i)_x\\
\ell_j(\s, \sigma) + \sigma(\s, v_i)_j &\geq \ell_j(\s, \sigma) + \sigma(\s, v_i)_x\\
\sigma(\s, v_i)_j &\geq \sigma(\s, v_i)_x \text.
\end{align*}
Thus, $f_j$ receives at least as much weight from $v_i$ as any other facility in the attraction range of $v_i$.
The proof for $f_p \in N_\s(v_i)$ with the highest facility load works analogously.
\end{proof}

With this, we prove that if we move a profile $\s$ from the \brokenmodel{} to the \feldmanmodel{} then the facility with the lowest load in the client equilibrium of the \brokenmodel{} has an equal or lower load in the \feldmanmodel{}.

\begin{lemma}
\label{lemma:min-broken}
In the \brokenmodel{}, given a facility placement profile $\s$ and a client equilibrium $\sigma$, let $f_j \in \mathcal{F}$ be the facility with the lowest load. Then $\ell_j(\s, \sigma) \geq \ell_j(\s, \sigma_\feldmanshort)$, where $\sigma_\feldmanshort$ is the client equilibrium for $\s$ in the \feldmanmodel{}.

Similarly, for the facility $f_p \in \mathcal{F}$ with the highest load, $\ell_p(\s, \sigma) \leq \ell_p(\s, \sigma_\feldmanshort{})$.
\end{lemma}
\begin{proof}
For an arbitrary client $v_i$, let $f_j$ be the facility with the lowest load in $N_\s(v_i)$. Thus by \Cref{lemma:min-receives-more}, $f_j$ receives as least as much weight from $v_i$ as any other facility in the range of $v_i$.
As only clients in $N_\s(v_i)$ receive weight from $v_i$, $f_j$ receives weight of at least $\sigma(\s, v_i)_j \geq \frac{w(v_i)}{\left|N_\s(v_i)\right|}$ from each $v_i$.
Thus, for each $v_i$, it holds that $\sigma(\s, v_i)_j \geq \sigma_\feldmanshort(\s, v_i)_j$ and therefore, $\ell_j(\s, \sigma) \geq \ell_j(\s, \sigma_\feldmanshort{})$.
The proof for $f_p \in \mathcal{F}$ with the highest facility load works analogously.
\end{proof}

We need another lemma to show that removing a facility from an instance of the \brokenmodel{} does not result in a utility loss for any other facility.

\begin{lemma}
\label{lemma:no-loss-removal}
Let $\s$ be a facility placement profile and $\sigma$ be a client equilibrium in the \brokenmodel{}.
If we remove a facility agent $f_j$ from the placement (and instance) resulting in $\s'$ with client equilibrium $\sigma'$, no facility $f_x \neq f_j$ loses utility.
\end{lemma}
\begin{proof}
Let $F_L$ be the set of facilities that lose utility by the removal of $f_j$.
We assume towards contradiction that this set is non-empty.
Since the sets of clients $A_\s(F_L)$ and $A_{\s'}(F_L)$ in the attraction range of $F_L$ are equal for both facility placement profiles, there must be some client $v_i \in A_\s(F_L)$ which allocates more weight outside $F_L$ in $\sigma'$ than in $\sigma$.
Thus, there exists a winning facility $f_w \notin F_L$, with $\ell_w(\s', \sigma') \geq \ell_w(\s, \sigma)$ and $\sigma'(\s', v_i)_w > \sigma(\s, v_i)_w$, and a losing facility $f_l \in F_L$, with $\ell_l(\s', \sigma') < \ell_l(\s, \sigma)$ and $\sigma'(\s', v_i)_l < \sigma(\s, v_i)_l$.
Thus, we get the two statements
\begin{align*}
\ell_w(\s', \sigma') + \sigma'(\s', v_i)_w &> \ell_w(\s, \sigma) + \sigma(\s, v_i)_w \text{ and}\\
\ell_l(\s', \sigma') + \sigma'(\s', v_i)_l &< \ell_l(\s, \sigma) + \sigma(\s, v_i)_l \text.
\end{align*}
By \Cref{lemma:shared-client}, we also get
\begin{align*}
\ell_w(\s, \sigma) + \sigma(\s, v_i)_w &\geq \ell_l(\s, \sigma) + \sigma(\s, v_i)_l \text{ and}\\
\ell_w(\s', \sigma') + \sigma'(\s', v_i)_w &\leq \ell_l(\s', \sigma') + \sigma'(\s', v_i)_l \text.
\end{align*}
Therefore, we arrive at a contradiction.
\end{proof}

We use the preceding lemmas to prove a 3-approximation:

\begin{theorem}
A $(1+\epsilon)$-approximate SPE in the \feldmanmodel{} is a $(3+2\epsilon)$-approximate SPE in the \brokenmodel{}.
\end{theorem}
\begin{proof}
In an SPE $(\s, \sigma_\feldmanshort{})$ in the \feldmanmodel{}, the facility $f_{l,\feldmanshort{}}$ with the lowest facility load and the facility $f_{h,\feldmanshort{}}$ with the highest facility load are separated by at most a factor of $2+2\epsilon$, as otherwise $f_{l,\feldmanshort{}}$ could improve by more than a factor of $1+\epsilon$ by deviating to the location of $f_{h,\feldmanshort{}}$ and thereby receive at least half her load.
When transferring $\s$ to the \brokenmodel{} with the corresponding client equilibrium $\sigma$, the factor between $f_{l,\brokenshort{}}$ with the lowest facility load and $f_{h,\brokenshort{}}$ and with the highest facility load is also at most $2+2\epsilon$, by \Cref{lemma:min-broken}.

Let an arbitrary facility $f_p$ make an improving move regarding the \brokenmodel{}, changing the facility placement profile from $\s$ to $\s'$.
Assume that $f_p$ has the highest utility $\ell_p(\s', \sigma_\brokenshort{})$ of all facilities in $(\s', \sigma_\brokenshort{})$ or
\begin{align}
    \label{eq:assumption-most}
    \ell_p(\s', \sigma_\brokenshort{}) = \max_{f_x \in \mathcal{F}}{\ell_x(\s', \sigma_\brokenshort{})}\text.
\end{align}
By \Cref{lemma:min-receives-more}, the facility $f_p$ receives at most a weight of $\sigma_\brokenshort(\s', v_i)_p \leq \frac{w(v_i)}{N_{\s'}(v_i)} = \sigma_\feldmanshort(\s', v_i)_p$ from each client $v_i \in A_{\s'}(f_p)$.
Thus, $\ell_p(\s', \sigma_\brokenshort{}) \leq \ell_p(\s', \sigma_\feldmanshort{}) \leq (2+2\epsilon)\ell_{l, \feldmanshort}(\s, \sigma_\feldmanshort{}) \leq (2+2\epsilon)\ell_p(\s, \sigma_\brokenshort{})$, where the last part holds by \Cref{lemma:min-broken}.
This means that if \Cref{eq:assumption-most} is true, the gain of facility $f_p$ is limited to a factor of $(2+2\epsilon)$.
Therefore, $f_p$ can only improve by a factor of more than $(2+2\epsilon)$ if another facility $f_x$ exists for which $\ell_p(\s', \sigma_\brokenshort{}) < \ell_x(\s', \sigma_\brokenshort{})$.

With the goal of finding an upper bound on $\ell_x(\s', \sigma_\brokenshort{})$, we investigate the move of $f_p$ in two parts:
First, the removal of $f_p$ resulting in $\s_r$ with the client equilibrium $\sigma_r$ in the \brokenmodel{}, second the reinsertion of $f_p$ in her new position.
By \Cref{lemma:no-loss-removal}, all facility utilities in $(\s_r, \sigma_r)$ have not decreased from the utilities in $(\s, \sigma_\brokenshort{})$.
Since the sum of utilities among the non-removed facilities increases by at most $\ell_p(\s, \sigma_\brokenshort{})$, the maximum utility gain of $f_x$ is at most
$$\ell_x(\s_r, \sigma_r) \leq \ell_x(\s, \sigma_\brokenshort{}) + \ell_p(\s, \sigma_\brokenshort{})\text.$$
By the inverse of \Cref{lemma:no-loss-removal}, adding $f_p$ in her new position cannot result in a utility gain for $f_x$ and so
$\ell_x(\s', \sigma_\brokenshort{}) \leq \ell_x(\s_r, \sigma_r)$.
Thus, we have
\begin{align*}
\ell_x(\s', \sigma_\brokenshort{}) &\leq \ell_x(\s, \sigma_\brokenshort{}) + \ell_p(\s, \sigma_\brokenshort{})\\
\ell_p(\s', \sigma_\brokenshort{}) &\leq \ell_x(\s, \sigma_\brokenshort{}) + \ell_p(\s, \sigma_\brokenshort{})\\
\ell_p(\s', \sigma_\brokenshort{}) &\leq (3+2\epsilon)\ell_p(\s, \sigma_\brokenshort{})
\end{align*}
concluding the proof.
\end{proof}

We will show that there exists a simple FPTAS to compute a $(1+\epsilon)$-approximate equilibrium in the \feldmanmodel{} which immediately yields the following theorem.
\begin{theorem}
 A $(3+2\epsilon)$-approximate SPE in the \brokenmodel{} can be computed in polynomial time.
\end{theorem}

As an algorithm to compute an $(1+\epsilon)$-approximate equilibrium in the \feldmanmodel{},  we employ approximate best response dynamics (see \Cref{alg:approx-br}).  Here we iteratively let facilities switch locations if they improve the payoff by a factor of at least $1+\epsilon$.


\begin{algorithm}[t]
    \caption{Approximate Best Response Dynamics}
    \label{alg:approx-br}
    $\s \gets$ arbitrary facility placement profile\;
    \While{ $\exists f_j \in \mathcal{F},$ with its best response $s_j' \in V$ and $\ell_j((s_j',s_{-j}),\sigma_\feldmanshort{}) \ge (1+\epsilon) \ell_j(\s,\sigma_\feldmanshort{})$}{
        $\s \gets (s_j',s_{-j})$\;
    }
\end{algorithm}

\begin{theorem}
There is a FPTAS to compute a $(1+\epsilon)$-approximate equilibrium in the \feldmanmodel{}.
\end{theorem}

\begin{proof}
By its stopping condition, \Cref{alg:approx-br} clearly computes a $(1+\epsilon)$-approximate equilibrium. As for the runtime, 
each improvement step can be performed in polynomial time iterating over all facilities and all their strategies.  Note that computing the cost of a player for each profile can also be done in polynomial time (cf. \Cref{theo:polytime-client}). 
Using the following lemma to bound the overall number of steps completes the proof.
\end{proof}
\begin{lemma}

Every sequence of $(1+\epsilon)$-best response steps in the \feldmanmodel{}
 converges in $\mathcal{O}(\frac{1}{\epsilon} k^2 \log k )$ steps.
\end{lemma}
\begin{proof}

For the \feldmanmodel{}, we have that $\Phi(\s)=\sum_{v \in V} \sum_{j=1}^{\left|N_\s(v)\right|} \frac{w(v)}{j}$ is an exact potential function that increases with each improving move of a facility exactly by the difference  of the improvement~\cite{Rosenthal}. That is,  if a facility $f_j$ improves from $\s$ by changing from $s_j$ to $s'_j$ with an improvement of $\Delta := \ell_j(\s,\sigma_\feldmanshort{}) - \ell_j((s_{-j},s'_j),\sigma_\feldmanshort{})$,  then $\Phi(\s) - \Phi(s_{-j},s'_j) = \Delta$.

We now prove the lemma by bounding the number of best response steps until we reach an approximate equilibrium. To that end,  let $\s^*$ be the equilibrium that maximizes the exact potential function $\Phi(\cdot)$. 

Note that an agent could always choose the location of the facility $f_p$ that covers the most client weight in $\s^*$.  That is $p = \arg\max_{\{1,\ldots,k\}} w(A_{\s^*}(f_j))$. By an averaging argument, that weight is at least $\frac{1}{k}$-th of the total weight $w(A_{\s^*}(\mathcal{F}))$ covered in $\s^*$.
Hence,  any best response of an agent yields a payoff of at least $\frac{1}{k}$-th the weight covered by $f_p$, which is  $\frac{1}{k}$ of the total load in $\s^*$.  So, the payoff of a best response is at least   $\frac{1}{k^2} w(A_{\s^*}(\mathcal{F}))$.

On the other hand,  $\Phi(\s^*)$ is at most $H_k$ times the total covered weight, i.e.,  $\Phi(\s^*) \le H_k w(A_{\s^*}(\mathcal{F}))$.  
Putting both together yields that the payoff of a best response is at least $\frac{1}{k^2 H_k} \Phi(\s^*)$.

As we are considering only best responses that increase a player's payoff by a factor of $1+\epsilon$,  every step improves the payoff (and,  hence,  the potential function) by at least 
$ \frac{\epsilon}{1+\epsilon} \frac{1}{k^2 H_k } \Phi(\s^*) $.

Therefore,  every sequence of $(1+\epsilon)$-best responses reaches an approximate equilibrium after at most $\mathcal{O}(\frac{1}{\epsilon} k^2 \log k )$ steps.
\end{proof}

Notably the $(1+\epsilon)$-factor is unavoidable,  as computing an exact equilibrium in the \feldmanmodel{} is PLS-complete.
\begin{theorem}
Computing an exact equilibrium in the \feldmanmodel{} is PLS-complete.
\end{theorem}
\begin{proof}
The problem is in PLS since we can compute  the potential function value for each facility placement profile in polynomial time and we can find a better solution in polynomial time by iterating over all unilateral deviations.

To prove hardness, we reduce from {\sc LocalMaxCut},  the local search version of {\sc MaxCut},  which is PLS-hard~\cite{schaffer1991simple,elsasser2011settling}.  That is, given a graph $G=(V,E)$ with edge weights $w_e$,  find a set $C \subseteq V$ such that the value of the cut,  i.e.,  $v(C):=\sum_{u\in C} \sum_{v \in V\setminus C} w_{(u,v)}$ cannot be improved by adding or removing one vertex to or from $C$. 

Given an instance of {\sc LocalMaxCut} with a graph $G=(V,E)$ with edge weights $w_e$,  we construct an instance of the \feldmanmodel{} on a host graph $H=(V',E',w')$ with $k=|V|$ facilities such that from an equilibrium in $H$ we can easily construct a local optimum of $G$.  For ease of exposition, we let $H$ be an undirected graph.  One can easily obtain an equivalent directed graph by duplicating edges.

For every vertex $v \in V$, there is a vertex gadget consisting of the five nodes {\em left}$_v$,  {\em right}$_v$,  {\em dummy}$1_v$,  {\em dummy}$2_v$,  and {\em dummy3}$_v$. The three dummy vertices have weight $M$,  the other two vertices have weight $0$.  There is an edge from {\em dummy1} to {\em left},  from {\em dummy2} to {\em right} and edges from {\em dummy3} to both,  {\em left} and {\em right}.

For every edge $e =(u,v) \in E$,  there are two edge vertices $v1_e$ and $v2_e$ each with weight $w_e$.  There is an edge from  $v1$ to the {\em right} vertex of the vertex gadget for $u$ and the {\em left} vertex of the vertex gadget of $v$.
 There is an edge from $v2$ to the {\em right} vertex of the vertex gadget for $v$ and the {\em left} vertex of the vertex gadget of $u$.
 
 As we have exactly $k=|V|$ facilities, it is easy to verify that in every equilibrium there is exactly one facility on either the {\em left} or the {\em right} vertex of each edge with a payoff of at least $2M$.  Note that more than one player in a vertex gadget or choosing a {\em dummy} or {\em edge} vertex gives significantly less payoff.
 
 It remains to show that we can determine a local optimum of the {\sc MaxCut} instance from any equilibrium in polynomial time. 
 We interpret an equilibrium profile as a {\sc LocalMaxCut} solution as follows: We define that  every vertex $v \in V$ where a facility is on the left node of the corresponding gadget is in $C$.  

In an equilibrium, the payoff of a player on a vertex gadget $u \in C$ is
$ 2M + \sum_{v \in \delta(u) \setminus C}w_{(u,v)}$.
For every player on a vertex gadget $u \not\in C$, the payoff is $2M +  \sum_{v \in \delta(u) \cap C} w_{(u,v)}$, where $\delta(u)$ is the set of neighboring nodes of $u$.
 
Since this is an equilibrium, deviating from the right to the left node within a gadget is not an improvement. Hence, the payoff for each player on a vertex gadget $u \in C$ is
$ 2M +  \sum_{v \in \delta(u) \setminus C}w_{(u,v)} \ge 
2M +  \sum_{v \in \delta(u) \cap C} w_{(u,v)}$.

Likewise, for each player on $u \not\in C$ the payoff is
$ 2M +  \sum_{v \in \delta(u) \cap C} w_{(u,v)} \ge
2M +  \sum_{v \in \delta(u) \setminus C}w_{(u,v)}$.
\end{proof}

\section{Conclusion}
We have shown that in our model of two-sided facility location subgame perfect equilibria are not always guaranteed to exist. This is in stark contrast to the model of \citet{KLMS21} in which clients exhibit a simpler behavior and merely perform load balancing. To resolve non-existence we studied approximate equilibria and showed the existence and polynomial time computability of approximate equilibria.

A major open problem is whether the approximation factors can be improved. We conjecture that $2$-approximate equilibria exist and that our approximation algorithm computes them.  On the negative side, a close inspection of our constructions in \Cref{theo:no-spe,theo:NP} shows already that there do not exist $\alpha$-approximate equilibria for a suitable small constant $\alpha$ and that the corresponding decision problem is intractable.  It would be very interesting to obtain a matching lower bound to the existence result.

\ifdefined\arxiv
\printbibliography
\else
\bibliography{results}
\fi

\ifdefined\arxiv
\appendix

\clearpage
\section*{Appendix}

\section{Existence of Client Equilibria}
\label{sec:existence}

We prove the existence of client equilibria with Kakutani's fixed-point theorem.
For that, we first show that a client always has just one unique best response.

\begin{lemma}
\label{lemma:unique-br}
A client $v_i$ has a unique best response for a given client distribution $\sigma$.
\end{lemma}
\begin{proof}
Given a client distribution $\sigma$, let the two responses $p=((p_1, \dots, p_k), \sigma_{-i})$ and $q=((q_1, \dots, q_k), \sigma_{-i})$ of client $v_i$ have the same cost.
Let the response $r=((r_1, \dots, r_{k}), \sigma_{-i})$, with $r_i = \frac{p_i + q_i}2$ be the midpoint of $p$ and $q$.
Since the strategy space of $v_i$ is a simplex, $r$ is a valid strategy.
The cost of client $v_i$ in $p$ (and analogously for q) is
\begin{align*}
L_i(\s,p) &= \sum_{j=1}^k p(\s,v_i)_j \ell_j(\s,p) \\
&= \sum_{j=1}^k p(\s,v_i)_j (p(\s,v_i)_j + \ell_j(\s,\sigma_{-i})) \\
&= \sum_{j=1}^k \left(p(\s,v_i)_j^2 +p(\s,v_i)_j \ell_j(\s,\sigma_{-i})\right)\text.
\end{align*}
For better readability we use $p_{i,j} = p(\s,v_i)_j$ and $q_{i,j} = q(\s,v_i)_j$ for the rest of this proof.
Now, we limit the cost of client $v_i$ at the midpoint $r$ to
\begin{align*}
L_i(\s,r) &= \sum_{j=1}^k \left(r(\s,v_i)_j^2 +r(\s,v_i)_j \ell_j(\s,\sigma_{-i})\right)\\
&= \sum_{j=1}^k \left(\frac{\left(p_{i,j} + q_{i,j}\right)^2}{4} +\frac{p_{i,j} + q_{i,j}}2 \ell_j(\s,\sigma_{-i})\right)\\
&< \sum_{j=1}^k \left(\frac{2p_{i,j}^2 + 2q_{i,j}^2}{4} +\frac{p_{i,j} + q_{i,j}}2 \ell_j(\s,\sigma_{-i})\right)\\
&= \frac{L_i(\s,p) + L_i(\s,q)}{2}\text.
\end{align*}
The second step is due to $a^2 + b^2 > 2ab \to 2a^2 + 2b^2 > (a+b)^2$, for $a\neq b$, and since $p$ and $q$ differ in at least one component.
Thus, $r$ has a lower cost than $p$ and $q$ for $v_i$.
This means, that a client $v_i$ cannot have two distinct best responses for a given $\sigma_{-i}$, as the midpoint of these best responses would have lower cost.
\end{proof}

Then, we prove the following statement with Kakutani's fixed-point theorem.
\begin{theorem}
A client equilibrium exists in all instances for all facility placement profiles $\s$.
\end{theorem}
\begin{proof}
The strategy space $S_i$ of a single client $v_i$ is an $|N(v_i)|$-simplex, where $N(v_i)$ is the the set of facilities in her shopping range.
Then the set $S = S_1 \times \dots \times S_n$ is non-empty, compact and convex.
We define the function $\phi: S \to 2^S$ to be $\phi(x) = \text{BR}_1(x_{-1}) \times \dots \times \text{BR}_n(x_{-n})$, where $\text{BR}_i(x_{-i})$ is the set of best responses of client $i$, given the facility placement profile $x_{-i}$ of all other clients.
By \Cref{lemma:unique-br}, $|BR_i(x_{-i})| = 1$ for all possible input.
Therefore, it also holds that $|\phi(x)| = 1$ for all $x$ and thus, $\phi(x)$ is always non-empty and convex.
$\phi$ has a closed graph, since it is continuous, as the best response of each player is the minimum of a quadratic function.
According to Kakutani's fixed-point theorem, $\phi$ must therefore have a fixed point $z$.
Since in $z$ all clients play there best response, $z$ is a client equilibrium.
\end{proof}

\section{Uniqueness of Client Equilibria}
\label{sec:uniqueness}
To prove the uniqueness of Nash equilibria in our client stage, we reduce it to an \emph{atomic splittable routing game}. 
However, we use the definition by \citet{bhaskar2015uniqueness} including the extension that the authors define in Section 5 of their paper, that each player may have her own individual delay function for each edge.
In that way, we can prohibit the use of edges that correspond to facilities not in the shopping range of a client by setting sufficiently bad delay functions.
For the reduction, we first give a definition of the atomic splittable routing game:
\begin{definition}[Atomic Splittable Routing Game \cite{bhaskar2015uniqueness}]
Given a graph $G=(V,E)$, let there be $k$ players with each player $i$ defined by $(w_i, s_i, t_i)$.
Each player $i$ routes $w_i$ units of flow from $s_i \in V$ to $t_i \in V$ which she can split arbitrarily among all $s_i$-$t_i$-paths.
Let $f$ be the total flow in the graph, $f_e$ the amount of flow through an edge $e$ and $f_e^i$ the amount of flow that player $i$ sends through edge $e$.
For each combination of player $i$ and edge $e$, there is a delay function $l_e^i(f_e)$.
The cost of player $i$ is $\sum_{e\in E}{f_e^il_e^i(f_e)}$.
\end{definition}
Now we prove that the \brokenmodel{} is isomorphic to a game that fulfills the conditions that are necessary to prove uniqueness of client equilibria.
\begin{lemma}
For a given facility placement profile $\s$ the client stage of the \brokenmodel{} is isomorphic to an atomic splittable routing game with a generalized nearly parallel graph and nonnegative, nondecreasing, differentiable, and convex delay functions.
\end{lemma}
\begin{proof}
We assume that all clients have at least one facility within their shopping range, as otherwise we can remove such clients without affecting the rest of the game.
We define $G'=(V',E')$ with $V'=\{s, t\} \cup \{\mathcal{F}\}$ and have two edges $(s,f_j)$ and $(f_j, t)$, for each facility $f_j \in \mathcal{F}$.
For each client $v_i$ we add a flow from $s$ to $t$ with weight $w(v_i)$, as $(w_i, s_i, t_i) = (w(v_i), s, t)$.
For each pair of a client $v_i$ and a facility $f_j$ with $f_j \in N_\s(v_i)$, we set the delay function of edge $(s, f_j)$ to $l_{(s, f_j)}^i(x)=x$.
For each client $v_i \in V$ and each facility $f_j \in \mathcal{F}$ we set $l_{(f_j, t)}^i(x)=0$.
Otherwise, we set $l_e^i(x)=x+3w_\s(\mathcal{F})$.
All functions $l_e^i$ are nonnegative, nondecreasing, differentiable, and convex and the graph $G'$ is generalized nearly parallel.

Next, we assume towards contradiction that a client $v_i$ puts a weight of $\epsilon$ on path $(s, f_j, t)$ with $f_j \notin N_\s(v_i)$.
Let $(s, f_p, t)$ be a path, with $f_p \in N_\s(v_i)$.
Since the second edge of each path has a delay of 0, the cost of $v_i$ for these two paths are $C = f^i_{(s, f_j)}l_{(s, f_j)}^i(f_{(s, f_j)}) + f^i_{(s, f_p)}l_{(s, f_p)}^i(f_{(s, f_p)}) = \epsilon(3w_\s(\mathcal{F}) + f_{(s, f_j)}) + f^i_{(s, f_p)}f_{(s, f_p)}$.
By moving this weight of $\epsilon$ to $(s, f_p, t)$, her costs are $C' = 0 \cdot l_{(s, f_j)}^i(f_{(s, f_j)} - \epsilon) + (\epsilon + f^i_{(s, f_p)})l_{(s, f_p)}^i(f_{(s, f_p)}+\epsilon) = (\epsilon + f^i_{(s, f_p)})(f_{(s, f_p)}+\epsilon) = \epsilon (\epsilon + f^i_{(s, f_p)} + f_{(s, f_p)}) + f^i_{(s, f_p)}f_{(s, f_p)} < \epsilon(3w_\s(\mathcal{F})) + f^i_{(s, f_p)}f_{(s, f_p)}\leq C$.
Thus, a client can always decrease her cost by not using a path corresponding to a facility not in her shopping range.
Therefore, the clients have isomorphic strategy spaces.
The cost functions in both games are also isomorphic since the second edge of each path has no delay and thus, the games are isomorphic.
\end{proof}

Now, we apply a theorem by \citet{bhaskar2015uniqueness}.
Note that $\mathscr{L}$ refers to the class of functions that are nonnegative, nondecreasing, differentiable, and convex.

\begin{theorem}[Theorem 6 by \citet{bhaskar2015uniqueness}]
Let $(G, \{(v_1, s_1, t_1), (v_2, s_2, t_2), \dots, (v_k, s_k, t_k)\}, l)$ be an atomic splittable routing game, where $l_e^i \in \mathscr{L}$ $\forall e \in E$. If graph $G$ is a generalized nearly parallel graph, then there is a unique equilibrium.
\end{theorem}

With that, we get the uniqueness of our client stage.
\uniqueness*

\section{Computation of Client Equilibria}
\label{sec:computation}

To compute client equilibria for a given facility placement profile $\s$, we reduce to \emph{atomic splittable singleton congestion games} by \citet{harks_equilibrium_2021}.

\begin{definition}[Atomic Splittable Singleton Congestion Game~\cite{harks_equilibrium_2021}]
An atomic splittable congestion game is a tuple $\mathcal{G}=\left(N,E, (d_i)_{i \in N}, (E_i)_{i \in N}, (c_{i,e})_{i \in N, e \in {E_i}}\right)$, with a set of players $N=\{1, \dots,n\}$ and a set of resources $E=\{e_1,\dots, e_m\}$.
Each player $i \in N$ has a weight $d_i \in \mathbb{Q}_{\geq 0}$ and a set of allowable resources $E_i \subseteq E$.
A strategy for a player $i \in N$ is a distribution of the weight $d_i$ over her allowable resources $E_i$, with $x_{i,e}$ being the weight of player $i$ on resource $e$.
The load $x_e = \sum{i \in N}{x_{i,e}}$ of a resource $e \in E$ is the sum of weights of the players on that resource.
The cost of player $i \in N$ is $\pi_i(x)= \sum_{e \in E_i}{c_{i,e}(x_e)x_{i,e}}$, with $c_{i,e}(x_e)=a_{i,e}x_e + b_{i,e}$ being a player-specific affine cost function, where $a_{i,e} \in \mathbb{Q}_{> 0}$ and $b_{i,e} \in \mathbb{Q}_{\geq 0}$.
\end{definition}

For the reduction, we use the facilities as resources in the atomic splittable singleton congestion game and connect them with the players according to their shopping ranges.

\begin{lemma}
\label{lem:harks-reduction}
The client stage of the \brokenmodel{} is isomorphic to an atomic splittable singleton congestion game~\cite{harks_equilibrium_2021}.
\end{lemma}
\begin{proof}
Let $N = V$, let $E=\mathcal{F}$, and for each $v_i \in V$, let $E_i = N_\s(v_i)$ and $d_i = w(v_i)$.
Furthermore, we set $a_{i,e} = 1$ and $b_{i,e} = 0$, for each $v_i \in V, e \in \mathcal{F}$, such that the cost of a player is $\pi_i(x) = x_e x_{i,e}$.
Since for these sets of parameters, the set of players (clients), the set of resources (facilities), the cost functions and the strategy spaces are exactly the same, both games are isomorphic to each other.
\end{proof}

Since our client game is isomorphic to the atomic splittable singleton congestion game, the algorithm by \citet{harks_equilibrium_2021} to compute an equilibrium applies to our game as well.
Note that $\delta$ is an upper bound on the maximum weight of the players and $k_0$ is the packet size for a discretized version of their game, for which the support sets (i.e., the pairs of clients $v_i$ and facilities $f_j$ with non-zero weight $\sigma(\s, v_i)_j$ in a client equilibrium $\sigma$) are equal to the original game.
Most importantly, $\log (\delta/k_0)$ is polynomial in the input size.
Additionally, in their model $n$ is the number of clients and $m$ is the number of resources, i.e., facilities in our model.

\begin{theorem}[Theorem 8 by \citet{harks_equilibrium_2021}]
Given game $\mathcal{G}$, we can compute an atomic splittable equilibrium for $\mathcal{G}$ in running time: $\mathcal{O}\left(
(nm)^3 + n^2m^{14} \log (\delta/k_0)\right)$.
\end{theorem}

With the same algorithm we may compute client equilibria in our game.

\computation*
\fi

\end{document}